\documentclass[a4paper,11pt,aps,pra,notitlepage,superscriptaddress,tightenlines]{revtex4-1}

\usepackage[utf8]{inputenc}

\usepackage[a4paper]{geometry}

\usepackage{amsfonts}
\usepackage{amsmath}
\usepackage{amssymb}
\usepackage{amsthm}  
\usepackage{mathrsfs}
\usepackage{nicefrac}
\usepackage{graphicx}

\usepackage[colorlinks=true,linkcolor=black,citecolor=black,plainpages=false,pdfpagelabels]{hyperref}
\hypersetup{pdftitle={Template}}
\usepackage{color}

\theoremstyle{plain}
\newtheorem{theorem}{Theorem}
\newtheorem{lemma}[theorem]{Lemma}

\newtheorem{corollary}[theorem]{Corollary}

\DeclareMathOperator{\tr}{Tr}
\newcommand{\id}{1}
\newcommand{\cP}{\mathcal{P}}
\newcommand{\cS}{\mathcal{S}}

\begin{document}

\title{\LARGE Relating different quantum generalizations of the conditional R\'enyi entropy}

\author{Marco Tomamichel}
\affiliation{Centre for Quantum Technologies, National University of Singapore, Singapore 117543, Singapore}
\affiliation{School of Physics, The University of Sydney, Sydney 2006, Australia}
\author{Mario Berta}
\affiliation{Institute for Quantum Information and Matter, Caltech, Pasadena, CA 91125, USA}
\affiliation{Institute for Theoretical Physics, ETH Zurich, 8092 Z\"urich, Switzerland}
\author{Masahito Hayashi}
\affiliation{Graduate School of Mathematics, Nagoya University, 
Furocho, Chikusaku, Nagoya, 464-860, Japan}
\affiliation{Centre for Quantum Technologies, National University of Singapore, Singapore 117543, Singapore}


\begin{abstract}
  Recently a new quantum generalization of the R\'enyi divergence and the corresponding conditional R\'enyi entropies was proposed. Here we report on a surprising relation between conditional R\'enyi entropies based on this new generalization and conditional R\'enyi entropies based on the quantum relative R\'enyi entropy that was used in previous literature. Our result  generalizes the well-known duality relation $H(A|B) + H(A|C) = 0$ of the conditional von Neumann entropy for tripartite pure states to R\'enyi entropies of two different kinds.
   As a direct application, we prove a collection of inequalities that relate different conditional R\'enyi entropies and derive a new entropic uncertainty relation.
\end{abstract}

\maketitle

\section{Introduction}

Recently, there has been renewed interest in finding suitable quantum generalizations of R\'enyi's~\cite{renyi61} entropies and divergences. This is due to the fact that R\'enyi entropies and divergences have a wide range of applications in classical information theory and cryptography, see, e.g.~\cite{csiszar95}.

We will review some of the recent progress here, but refer the reader to~\cite{lennert13} for a more in-depth discussion. 
For our purposes, a quantum system is modeled by a finite dimensional Hilbert space. We denote by $\cP$ the set of positive semi-definite operators on that Hilbert space, and by $\cS$ the subset of density operators with unit trace.

The following natural quantum generalization of the Rényi divergence has been widely used and has found operational significance, for example, as a cut-off rate in quantum hypothesis testing~\cite{mosonyi11} (see also~\cite{ogawa04,nagaoka06}). It is usually referred to as \emph{quantum R\'enyi relative entropy} and for all $\alpha \in (0, 1) \cup (1, \infty)$ given as
\begin{equation}
  D_{\alpha}(\rho\|\sigma) := \frac{1}{\alpha-1} \log\tr\left\{\rho^{\alpha}\sigma^{1-\alpha}\right\} \label{eq:dold}
\end{equation}
for arbitrary $\rho \in \cS$, $\sigma \in \cP$ that satisfy $\rho \ll \sigma$. (The notation $\rho \ll \sigma$ means that $\sigma$ dominates $\rho$, i.e.\ the kernel of $\sigma$ lies inside the kernel of $\rho$.)

While this definition has proven useful in many applications, it has a major drawback in that it does not satisfy the data-processing inequality (DPI) for $\alpha > 2$. The DPI states that the quantum R\'enyi relative entropy is contractive under application of a quantum channel, i.e., $D_{\alpha}\big(\mathcal{E}[\rho]\big\|\mathcal{E}[\sigma]\big) \leq D_{\alpha}(\rho\|\sigma)$ for any completely positive trace-preserving map $\mathcal{E}$. Intuitively, this property is very desirable since we want to think of the divergence as a measure of how well $\rho$ can be distinguished from $\sigma$, and this can only get more difficult after a channel is applied.

Recently, an alternative quantum generalization has been investigated~\cite{martinthesis,lennert13,wilde13} (see also~\cite{mytutorial12}). It is referred to as \emph{quantum R\'enyi divergence} (or sandwiched R\'enyi relative entropy in~\cite{wilde13}) and defined as
\begin{equation}
  \widetilde{D}_{\alpha}(\rho\|\sigma) := \frac{1}{\alpha-1} \log\tr\left\{ \left( \sigma^{\frac{1-\alpha}{2\alpha}} \rho \sigma^{\frac{1-\alpha}{2\alpha}} \right)^{\alpha} \right\} \label{eq:dnew}
\end{equation}
for all $\alpha \in (0, 1) \cup (1, \infty)$ and $\rho \in \cS$, $\sigma \in \cP$ that satisfy $\rho \ll \sigma$.
The quantum R\'enyi divergence has found operational significance in the converse part of quantum hypothesis testing~\cite{mosonyiogawa13}. As such, it satisfies the DPI for all $\alpha \geq \frac12$ as was shown by Frank and Lieb~\cite{frank13} and independently by Beigi~\cite{beigi13} for $\alpha > 1$. See also earlier work~\cite{martinthesis,lennert13} where a different proof is given for $\alpha \in (1, 2]$. Furthermore, the quantum R\'enyi divergence has already proven an indispensable tool, for example in the study of strong converse capacities of quantum channels~\cite{wilde13,Gupta13}.

The definitions,~\eqref{eq:dold} and~\eqref{eq:dnew}, are in general different but coincide when $\rho$ and $\sigma$ commute. 
For $\alpha \in \{0, 1, \infty\}$, we define $D_{\alpha}(\rho\|\sigma)$ and $\widetilde{D}_{\alpha}(\rho\|\sigma)$ as the corresponding limit. For $\alpha\to 0$ it has been shown that~\cite{datta13,audenaert13}:
\begin{align}
D_{0}(\rho\|\sigma) &=-\log\tr\left\{\Pi_{\rho}\sigma\right\}\,\\
\widetilde{D}_{0}(\rho\|\sigma) &=-\log\max_{i_{1},\ldots,i_{s}}\left\{\sum_{j=1}^{s}\lambda_{i_{j}}:\left\{\Pi_{\rho}|i_{j}\rangle\right\}\;\mathrm{linearly}\;\mathrm{independent}\right\}\,
\end{align}
with the eigenvalue decomposition $\sigma=\sum_{i}\lambda_{i}|i\rangle\langle i|$, $s=\mathrm{rank}\left(\Pi_{\rho}\sigma\right)$, and $\Pi_{\rho}$ the projector on the support of $\rho$. In the limit $\alpha \to 1$ both expressions converge to the quantum relative entropy~\cite{martinthesis,lennert13,wilde13}, namely
\begin{align}
  D_{1}(\rho\|\sigma) = \widetilde{D}_{1}(\rho\|\sigma) = D(\rho\|\sigma) := \tr\big\{\rho(\log \rho - \log \sigma)\big\} \,.
\end{align}
For $\alpha\to\infty$ the limits have been evaluated in~\cite{lennert13} and~\cite{tomamichel08}, respectively:
\begin{align}
\widetilde{D}_{\infty}(\rho\|\sigma) &= \inf \big\{ \lambda \in \mathbb{R} : \rho\leq 2^{\lambda} \sigma\big\}\,\\
  D_{\infty}(\rho\|\sigma) &=\log\max_{i,j}\left\{\frac{\nu_{i}}{\mu_{j}}:\langle i|\bar{j}\rangle\neq0\right\}\,
\end{align}
with the eigenvalue decompositions $\rho=\sum_{i}\nu_{i}|i\rangle\langle i|$ and $\sigma=\sum_{j}\mu_{j}|\bar{j}\rangle\langle\bar{j}|$. 

It has been observed~\cite{wilde13,datta13} that the relation
\begin{align}
  D_{\alpha}(\rho\|\sigma) &\geq \widetilde{D}_{\alpha}(\rho\|\sigma) 
  \label{eq:lieb-trace}
\end{align}
follows from the Araki-Lieb-Thirring trace inequality~\cite{ariki90,liebthirring05}.
Furthermore, $\alpha \mapsto D_{\alpha}(\rho\|\sigma)$ and $\alpha \mapsto \widetilde{D}_{\alpha}(\rho\|\sigma)$ are monotonically increasing functions. For the latter quantity, this was shown in~\cite{lennert13} and independently in~\cite{beigi13}.

Finally, very recently Audenaert and Datta~\cite{audenaert13} defined a more general two parameter family of \emph{$\alpha$-z-relative R\'enyi entropies} of the form
\begin{align}
D_{\alpha,z}(\rho\|\sigma):=\frac{1}{\alpha-1}\log\tr\left\{\left(\rho^{\frac{\alpha}{z}}\sigma^{\frac{1-\alpha}{z}} \right)^{z}\right\} ,
\end{align}
and explored some of its properties. We clearly have
$D_{\alpha} \equiv D_{\alpha,1}$ and $\widetilde{D}_{\alpha} \equiv D_{\alpha,\alpha}$.

\section{Quantum Conditional R\'enyi Entropies}

We will in the following consider disjoint quantum systems, denoted by capital letters $A, B$ and $C$. The sets $\cP(A)$ and $\cS(A)$ take on the expected meaning.

The conditional von Neumann entropy can be conveniently defined in terms of the quantum relative entropy as follows. For a bipartite state $\rho_{AB} \in \cS(AB)$, we define
\begin{align}
  H(A|B)_{\rho} :\!&= H(\rho_{AB}) - H(\rho_B) \label{eq:vn}\\
  &= - D(\rho_{AB} \| \id_A \otimes \rho_B) \label{eq:vn1}\\
  &= \sup_{\sigma_B \in \cS(B)} - D(\rho_{AB} \| \id_A \otimes \sigma_B),\label{eq:vn2}
\end{align}
where $H(\rho) := - \tr\{\rho \log \rho\}$ is the usual von Neumann entropy. The last equality can be verified using the relation $D(\rho_{AB}\|\id_A \otimes \sigma_B) = D(\rho_{AB}\|\id_A \otimes \rho_B) + D(\rho_B\|\sigma_B)$ together with the fact that $D(\cdot\|\cdot)$ is positive definite.

In the case of R\'enyi entropies, it is not immediate which expression, \eqref{eq:vn}, \eqref{eq:vn1} or~\eqref{eq:vn2}, should be used to define the conditional R\'enyi entropies. It has been found in the study of the classical special case (see, e.g.~\cite{Iwamoto13} for an overview) that generalizations based on~\eqref{eq:vn} have severe limitations, for example they cannot be expected to satisfy a DPI.
On the other hand, definitions based on the underlying divergence, as in~\eqref{eq:vn1} or~\eqref{eq:vn2}, have proven to be very fruitful and lead to quantities with operational significance.
Together with the two proposed quantum generalizations of the R\'enyi divergence in~\eqref{eq:dold} and~\eqref{eq:dnew},
this leads to a total of four different candidates for conditional R\'enyi entropies. For $\alpha \geq 0$ and $\rho_{AB} \in \cS(AB)$, we define
\begin{align}
  H_{\alpha}^{\downarrow}(A|B)_{\rho} &:= - D_{\alpha}(\rho_{AB}\|\id_A \otimes \rho_B), \label{eq:had} \\
  H_{\alpha}^{\uparrow}(A|B)_{\rho} &:= \sup_{\sigma_B \in \cS(B)} - D_{\alpha}(\rho_{AB}\|\id_A \otimes \sigma_B), \label{eq:hau} \\
    \widetilde{H}_{\alpha}^{\downarrow}(A|B)_{\rho} &:= - \widetilde{D}_{\alpha}(\rho_{AB}\|\id_A \otimes \rho_B), \qquad \qquad \qquad \textrm{and} \label{eq:hatd} \\
  \widetilde{H}_{\alpha}^{\uparrow}(A|B)_{\rho} &:= \sup_{\sigma_B \in \cS(B)} - \widetilde{D}_{\alpha}(\rho_{AB}\|\id_A \otimes \sigma_B) \label{eq:hatu}.
\end{align}

The fully quantum entropy $H_{\alpha}^{\downarrow}$ has first been studied in~\cite{tomamichel08}. For the classical and classical-quantum special case this quantity gives a generalization of the leftover hashing lemma~\cite{bennett95} for the modified mutual information to R\'enyi entropies with $\alpha\neq2$~\cite{Hayashi11_2,hayashi12}.

The classical version of $H_{\alpha}^{\uparrow}$ was introduced by Arimoto for an evaluation of the guessing probability~\cite{Arimoto75}. We note that he used another but equivalent expression for $H_{\alpha}^{\uparrow}$ that we later explain in Lemma~\ref{lm:dau-new}. Then, Gallager used $H_{\alpha}^{\uparrow}$ (again in the form of Lemma~\ref{lm:dau-new}) to upper bound the decoding error probability of a random coding scheme for data compression with side-information~\cite{gallager79,Yagi12}. The classical and classical-quantum special cases of $H_{\alpha}^{\uparrow}$ were, for example, also investigated in~\cite{hayashi12,Hayashi13_2} and realize another type of a generalization of the leftover hashing lemma for the $L_1$-distinguishability in the study of randomness extraction to R\'enyi entropies with $\alpha\neq2$.

It follows immediately from the definition and the corresponding property of $D_{\alpha}$ that these two entropies satisfy a data-processing inequality. Namely for any quantum operation $\mathcal{E}_{B\rightarrow B'}$ with $\tau_{AB'} = \mathcal{E}_{B\rightarrow B'}[\rho_{AB}]$ and any $\alpha \in [0, 2]$, we have 
\begin{align}
  H_{\alpha}^{\downarrow}(A|B)_{\rho} \leq H_{\alpha}^{\downarrow}(A|B')_{\tau} \qquad \textrm{and} \qquad 
  H_{\alpha}^{\uparrow}(A|B)_{\rho} \leq H_{\alpha}^{\uparrow}(A|B')_{\tau} 
\end{align}
while their classical-quantum versions have been obtained in~\cite{hayashi12}.

The conditional entropy $\widetilde{H}_{\alpha}^{\uparrow}$ was proposed in~\cite{mytutorial12} and investigated in~\cite{lennert13}, whereas $\widetilde{H}_{\alpha}^{\downarrow}$ is first considered in this paper. (Since the relative entropies $\widetilde{D}_{\alpha}$ and $D_{\alpha}$ are identical for commuting operators, we note that $\widetilde{H}_{\alpha}^{\uparrow} = H_{\alpha}^{\uparrow}$ as well as $\widetilde{H}_{\alpha}^{\downarrow} = H_{\alpha}^{\downarrow}$ for classical distributions.) Both definitions satisfy the above data-processing inequality for $\alpha \geq \frac12$.

Furthermore, it is easy to verify that all entropies considered are invariant under applications of local isometries on either the $A$ or $B$ systems. Lastly, note that the optimization over $\sigma_B$ can always be restricted to $\sigma_B \gg \rho_B$ for $\alpha > 1$.

We use up and down arrows to express the trivial observation that $H_{\alpha}^{\uparrow}(A|B)_{\rho} \geq H_{\alpha}^{\downarrow}(A|B)_{\rho}$ and $\widetilde{H}_{\alpha}^{\uparrow}(A|B)_{\rho} \geq \widetilde{H}_{\alpha}^{\downarrow}(A|B)_{\rho}$ by definition. Finally,~\eqref{eq:lieb-trace} gives us the additional relations $\widetilde{H}_{\alpha}^{\uparrow}(A|B)_{\rho} \geq H_{\alpha}^{\uparrow}(A|B)_{\rho}$ and $\widetilde{H}_{\alpha}^{\downarrow}(A|B)_{\rho} \geq H_{\alpha}^{\downarrow}(A|B)_{\rho}$. These relations are summarized in Figure~\ref{fig:overview}. Moreover, inheriting these properties from the corresponding divergences, all entropies are monotonically decreasing functions of $\alpha$ 

For $\alpha = 1$, all definitions coincide with the usual von Neumann conditional entropy~\eqref{eq:vn1}. For $\alpha = \infty$, two quantum generalizations of the conditional min-entropy emerge, which both have been studied by Renner~\cite{renner05}. Namely,
\begin{align}
  \widetilde{H}_{\infty}^{\downarrow}(A|B)_{\rho} &= \sup \big\{ \lambda \in \mathbb{R} : \rho_{AB} \leq 2^{-\lambda} \id_A \otimes \rho_B \big\} \label{eq:min1} \qquad \textrm{and} \\
  \widetilde{H}_{\infty}^{\uparrow}(A|B)_{\rho} &= \sup \big\{ \lambda \in \mathbb{R} : \rho_{AB} \leq 2^{-\lambda} \id_A \otimes \sigma_B,\, \sigma_B \in \cS(B) \big\} \label{eq:min2} .
\end{align}
(The notation $H_{\min}(A|B)_{\rho|\rho} \equiv \widetilde{H}_{\infty}^{\downarrow}(A|B)_{\rho}$ and $H_{\min}(A|B)_{\rho} \equiv \widetilde{H}_{\infty}^{\uparrow}(A|B)_{\rho}$ is widely used. However, we prefer our notation as it makes our exposition in this manuscript clearer.)
For $\alpha=2$, we find a quantum generalization of the conditional collision entropy as introduced by Renner~\cite{renner05}:
\begin{align}
\widetilde{H}_{2}^{\downarrow}(A|B)_{\rho}=-\log\tr\left\{\left(\rho_{AB} \Big(1_{A}\otimes\rho_{B}^{-\frac12} \Big)\right)^{2}\right\}. \label{eq:h2}
\end{align}
For $\alpha= \frac12$, we find the quantum conditional max-entropy first studied by K\"onig \emph{et al.}~\cite{koenig08},
\begin{align}
\widetilde{H}_{\nicefrac{1}{2}}^{\uparrow}(A|B)_{\rho}=\sup_{\sigma_B \in \cS(B)} 2 \log F(\rho_{AB},1_{A}\otimes\sigma_{B})\,,\label{eq:hmax}
\end{align}
where $F(\cdot,\cdot)$ denotes the fidelity. (The alternative notation $H_{\max}(A|B)_{\rho} \equiv \widetilde{H}_{\nicefrac{1}{2}}^{\uparrow}(A|B)_{\rho}$ is often used.) For $\alpha = 0$, we find a quantum conditional generalization of the Hartley entropy~\cite{hartley28} that was initially considered by Renner~\cite{renner05},
\begin{align}
H_0^{\uparrow}(A|B)_{\rho} = \sup_{\sigma_B \in \cS(B)} \log \tr \{ \Pi_{\rho_{AB}}\, \id_A \otimes \sigma_B  \}\,,\label{eq:h0}
\end{align}
where $\Pi_\rho$ denotes the projector onto the support of $\rho$.

\begin{figure}
\includegraphics{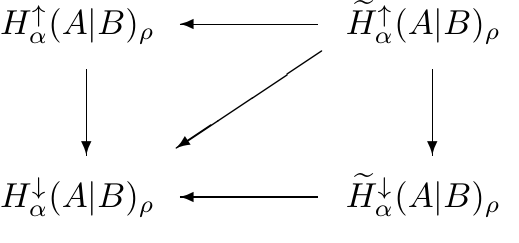}
\caption{Overview of the different conditional entropies used in this paper. Arrows indicate that one entropy is larger or equal to the other for all states $\rho_{AB} \in \cS(AB)$ and all $\alpha \geq 0$.}
\label{fig:overview}
\end{figure}

\section{Duality Relations}

It is well known that, for any tripartite pure state $\rho_{ABC}$, the relation
\begin{align}
  H(A|B)_{\rho} + H(A|C)_{\rho} = 0 \label{eq:dual-vn}
\end{align}
holds. We call this a \emph{duality relation} for the conditional entropy. To see this, simply write $H(A|B)_{\rho} = H(\rho_{AB}) - H(\rho_B)$ and $H(A|C)_{\rho} = H(\rho_{AC}) - H(\rho_C)$ and note that the spectra of $\rho_{AB}$ and $\rho_C$ as well as the spectra of $\rho_{B}$ and $\rho_{AC}$ agree. The significance of this relation is manifold\,---\,for example it turns out to be useful in cryptography where the entropy of an adversarial party, let us say $C$, can be estimated using local state tomography by two honest parties, $A$ and $B$.
In the following, we are interested to see if such relations hold more generally for conditional R\'enyi entropies.

It was shown in~\cite[Lem.~6]{tomamichel08} that $H_{\alpha}^{\downarrow}$ indeed satisfies a duality relation, namely
\begin{align}
  H_{\alpha}^{\downarrow}(A|B)_{\rho} + H_{\beta}^{\downarrow}(A|C)_{\rho} = 0 \qquad \textrm{when}\quad \alpha + \beta = 2,\ \alpha, \beta \geq 0 \,.
\end{align}
Note that the map $\alpha \mapsto \beta = 2 - \alpha$ maps the interval $[0, 2]$, where data-processing holds, onto itself. This is not surprising. Indeed, consider the Stinespring dilation $\mathcal{U}_{B \to B'B''}$ of a quantum channel $\mathcal{E}_{B \to B'}$. Then, for $\rho_{ABC}$ pure, $\tau_{AB'B''C} = \mathcal{U}_{B \to B'B''}[\rho_{ABC}]$ is also pure and the above duality relation implies that
\begin{align}
  H_{\alpha}^{\downarrow}(A|B)_{\rho} \leq H_{\alpha}^{\downarrow}(A|B')_{\tau} \iff H_{\beta}^{\downarrow}(A|C)_{\rho} \geq H_{\beta}^{\downarrow}(A|B''C)_{\tau} .
\end{align}
Hence, data-processing for $\alpha$ holds if and only if data-processing for $\beta$ holds.

A similar relation has recently been discovered for $\widetilde{H}_{\alpha}^{\uparrow}$ in~\cite{lennert13} and independently in~\cite{beigi13}. There, it is shown that
\begin{align}
  \widetilde{H}_{\alpha}^{\uparrow}(A|B)_{\rho} + \widetilde{H}_{\beta}^{\uparrow}(A|C)_{\rho} = 0 \qquad \textrm{when}\quad \frac1{\alpha} + \frac{1}{\beta} = 2,\ \alpha, \beta \geq \frac12 \,.
\end{align}
As expected, the map $\alpha \mapsto \beta = \frac{\alpha}{2\alpha-1}$ maps the interval $[\frac12,\infty]$, where data-processing holds, onto itself.

The purpose of the following is thus to show if a similar relation holds for the remaining two candidates, $H_{\alpha}^{\uparrow}$ and $\widetilde{H}_{\alpha}^{\downarrow}$. First, we find the following alternative expression for $H_{\alpha}^{\uparrow}$ by determining the optimal $\sigma_B$ in the definition~\eqref{eq:hau}.

\begin{lemma}
  \label{lm:dau-new}
  Let $\alpha \in (0, 1) \cup (1, \infty)$ and $\rho_{AB} \in \cS(AB)$. Then,
  \begin{align}
    H_{\alpha}^{\uparrow}(A|B)_{\rho} = \frac{\alpha}{1-\alpha} \log \tr\Big\{ \big( \tr_A \{ \rho_{AB}^{\alpha} \} \big)^{\frac{1}{\alpha}}  \Big\} . \label{eq:dau-new}
  \end{align}
\end{lemma}

This generalizes a result by one of the current authors~\cite[Lem.~7]{hayashi12}.

\begin{proof}
  Recall the definition
  \begin{align}
    H_{\alpha}^{\uparrow}(A|B)_{\rho} &= \sup_{\sigma_B \in \cS(B)} \frac{1}{1-\alpha} \log \tr\big\{ \rho_{AB}^{\alpha}\,\id_A \otimes \sigma_B^{1-\alpha} \big\} \\&
    = \sup_{\sigma_B \in \cS(B)} \frac{1}{1-\alpha} \log \tr\big\{ \tr_A \{ \rho_{AB}^{\alpha}\} \sigma_B^{1-\alpha} \big\}.
  \end{align}
  This can immediately be lower bounded by the expression in~\eqref{eq:dau-new} by substituting
  \begin{align}
    \sigma_B^* = \frac{\big( \tr_A \{ \rho_{AB}^{\alpha} \} \big)^{\frac{1}{\alpha}}}{\tr \Big\{ \big( \tr_A \{ \rho_{AB}^{\alpha} \} \big)^{\frac{1}{\alpha}} \Big\}}
    \end{align}
    for $\sigma_B$. It remains to show that this choice is optimal. We employ the following H\"older and reverse H\"older inequalities (cf.~Lemma~\ref{lm:hoelder} in Appendix~\ref{app:hoelder}). For any $A, B \geq 0$, the H\"older inequality states that 
    \begin{align}
      \tr\{ A B \} \leq \big( \tr \{ A^{p} \} \big)^{\frac{1}{p}} \big( \tr \{ B^{q} \} \big)^{\frac{1}{q}} \qquad \textrm{for all }  p, q > 1 \textrm{ s.t. }\frac{1}{p} + \frac{1}{q} = 1 . \label{eq:hh1}
    \end{align}
    Furthermore, if $B \gg A$, we also have a reverse H\"older inequality which states that
    \begin{align}
      \tr\{ A B \} \geq \big( \tr \{ A^{p} \} \big)^{\frac{1}{p}} \big( \tr \{ B^{q} \} \big)^{\frac{1}{q}} \qquad \textrm{for all }q < 0 < p < 1 \textrm{ s.t. } \frac{1}{p} + \frac{1}{q} = 1 . \label{eq:hh2}
    \end{align}
    For $\alpha < 1$, we employ~\eqref{eq:hh1} for $p = \frac{1}{\alpha}$, $q = \frac{1}{1-\alpha}$, $A = \tr_A\{\rho_{AB}^{\alpha} \}$ and $B = \sigma_B^{1-\alpha}$ to find
    \begin{align}
      \tr\big\{ \tr_A \{ \rho_{AB}^{\alpha}\} \sigma_B^{1-\alpha} \big\}
      \leq \bigg( \tr\Big\{ \big( \tr_A \{ \rho_{AB}^{\alpha}\} \big)^{\frac{1}{\alpha}} \Big\} \bigg)^{\alpha} \big( \tr\{\sigma_B\} \big)^{1-\alpha},
    \end{align}
    which yields the desired upper bound since $\tr\{\sigma_B\} = 1$.
    For $\alpha > 1$, we instead use~\eqref{eq:hh2}. 
    This leads us to~\eqref{eq:dau-new} upon the same substitutions, concluding the proof.
\end{proof}

An alternative proof also follows rather directly from a quantum generalization of Sibson's identity, which was introduced by Sharma and Warsi~\cite[Lem.~3 in Suppl.~Mat.]{sharma13}.

This allows us to show our main result.
\begin{theorem}
  Let $\alpha,\beta \in (0, 1) \cup (1, \infty)$ with $\alpha \cdot \beta = 1$ and let $\rho_{ABC} \in \cS(ABC)$ be pure. Then,
    $H_{\alpha}^{\uparrow}(A|B)_{\rho} + \widetilde{H}_{\beta}^{\downarrow}(A|C)_{\rho} = 0$.
\end{theorem}

\begin{proof}
  Substituting $\beta = \frac{1}{\alpha}$ and employing Lemma~\ref{lm:dau-new}, it remains to show that
  \begin{align}
    &H_{\alpha}^{\uparrow}(A|B)_{\rho} = \frac{\alpha}{1-\alpha} \log \tr \Big\{ \big( \tr_A \{ \rho_{AB}^{\alpha} \} \big)^{\frac{1}{\alpha}}  \Big\}
  \end{align}
  is equal to
  \begin{align}
    -\widetilde{H}_{\beta}^{\downarrow}(A|C)_{\rho} &= - \frac{1}{1-\beta} \log \tr \Bigg\{ \bigg( \Big(\id_A \otimes \rho_C^{\frac{1-\beta}{2\beta}}\Big) \rho_{AC} \Big(\id_A \otimes \rho_C^{\frac{1-\beta}{2\beta}}\Big) \bigg)^{\beta} \Bigg\} \\
    &= \frac{\alpha}{1-\alpha} \log \tr \Bigg\{ \bigg( \Big(\id_A \otimes \rho_C^{\frac{\alpha-1}{2}} \Big) \rho_{AC} \Big( \id_A \otimes \rho_C^{\frac{\alpha-1}{2}} \Big) \bigg)^{\frac{1}{\alpha}}  \Bigg\} .
  \end{align}
  In the following we show something stronger, namely that the operators 
  \begin{equation}
   \tr_A \{ \rho_{AB}^{\alpha} \} \qquad \textrm{and} \qquad
   \Big(\id_A \otimes \rho_C^{\frac{\alpha-1}{2}}\Big) \rho_{AC} \Big(\id_A \otimes \rho_C^{\frac{\alpha-1}{2}} \Big) \label{eq:marginals}
   \end{equation}
   are unitarily equivalent. This is true since both of these operators are marginals\,---\,on $B$ and $AC$\,---\,of the same tripartite rank-$1$ operator, 
   \begin{align}
     \Big( \id_{AB} \otimes \rho_{C}^{\frac{\alpha-1}{2}} \Big) \rho_{ABC} \Big( \id_{AB} \otimes \rho_{C}^{\frac{\alpha-1}{2}} \Big) .
   \end{align}
    To see that this is indeed true, note the first operator in~\eqref{eq:marginals} can be rewritten as
   \begin{align}
     \tr_A \{ \rho_{AB}^{\alpha} \} &= \tr_A \Big\{ \rho_{AB}^{\frac{\alpha-1}{2}} \rho_{AB}\, \rho_{AB}^{\frac{\alpha-1}{2}} \Big\} \\
     &= \tr_{AC} \Big\{ \Big( \rho_{AB}^{\frac{\alpha-1}{2}} \otimes \id_C \Big) \rho_{ABC} \Big( \rho_{AB}^{\frac{\alpha-1}{2}} \otimes \id_C \Big) \Big\} \\
     &= \tr_{AC} \Big\{ \Big( \id_{AB} \otimes \rho_{C}^{\frac{\alpha-1}{2}} \Big) \rho_{ABC} \Big( \id_{AB} \otimes \rho_{C}^{\frac{\alpha-1}{2}} \Big) \Big\} \, .
   \end{align}
   The last equality can be verified using the Schmidt decomposition of $\rho_{ABC}$ with regards to the partition $AB$:$C$.
   This concludes the proof.
\end{proof}
The relation can readily be extended for all $\alpha \geq 0$ and $\beta \geq 0$. The limiting case $\alpha = 1$ is simply the duality of the conditional von Neumann entropy~\eqref{eq:dual-vn}, whereas the case $\alpha = 0, \beta = \infty$ was also shown in~\cite[Prop.~3.11]{berta08}. (See~\cite[Lem.~25]{tomamichel10} for a concise proof.) 
Again, note that the transformation $\alpha \mapsto \beta = \frac{1}{\alpha}$ maps the interval $[0, 2]$ where data-processing holds for $H_{\alpha}^{\uparrow}$ to $[\frac{1}{2}, \infty]$ where data-processing holds for $\widetilde{H}_{\beta}^{\downarrow}$.

We summarize these duality relations in the following theorem, where we take note that the first and second statements have been shown in~\cite{tomamichel08} and~\cite{lennert13,beigi13}, respectively.
\begin{theorem}
  \label{thm:dual}
  For any pure $\rho_{ABC} \in \cS(ABC)$, the following holds:\footnote{We use the convention that $\frac{1}{\infty} = 0$ and $0\cdot\infty = 1$.}
  \begin{align}
    H_{\alpha}^{\downarrow}(A|B)_{\rho} + H_{\beta}^{\downarrow}(A|C)_{\rho} &= 0 \qquad && \textrm{for}  &&\alpha, \beta \in [0, 2],\, \alpha + \beta = 2, \\
   \widetilde{H}_{\alpha}^{\uparrow}(A|B)_{\rho} + \widetilde{H}_{\beta}^{\uparrow}(A|C)_{\rho} &= 0 && \textrm{for} &&\alpha, \beta \in \Big[\frac12, \infty\Big],\, \frac{1}{\alpha} + \frac{1}{\beta} = 2, \\
   H_{\alpha}^{\uparrow}(A|B)_{\rho} + \widetilde{H}_{\beta}^{\downarrow}(A|C)_{\rho} &= 0 && \textrm{for} &&\alpha, \beta \in [0, \infty],\, \alpha \cdot \beta = 1 . \label{eq:dual3}
\end{align}
\end{theorem}

\section{Some Inequalities Relating Conditional Entropies}

Our first application yields relations between different conditional R\'enyi entropies for arbitrary mixed states.
Recently, Mosonyi~\cite[Lem.~2.1]{mosonyi13} used a converse of the Araki-Lieb-Thirring trace inequality due to Audenaert~\cite{Audenaert08} to find a converse to the ordering relation $D_{\alpha}(\rho\|\sigma) \geq \widetilde{D}_{\alpha}(\rho\|\sigma)$, namely
\begin{align}
\widetilde{D}_{\alpha}(\rho\|\sigma) \geq \alpha\cdot D_{\alpha}(\rho\|\sigma)+ \log\tr \big\{\rho^{\alpha} \big\} +(\alpha - 1)\log\|\sigma\| \,.
\end{align}
Here we follow a different approach and show that inequalities of a similar type for the conditional entropies are a direct corollary of the duality relations in Theorem~\ref{thm:dual}.

\begin{corollary}
  Let $\rho_{AB} \in \cS(AB)$. Then, the following inequalities hold for $\alpha \in \left[\frac{1}{2}, \infty\right]$:
  \begin{align}
    H_{\alpha}^{\uparrow}(A|B)_{\rho} \leq \widetilde{H}_{\alpha}^{\uparrow}(A|B)_{\rho} &\leq H_{2 - \frac{1}{\alpha}}^{\uparrow}(A|B)_{\rho}\,,\label{eq:ineq1}\\
    H_{\alpha}^{\downarrow}(A|B)_{\rho} \leq H_{\alpha}^{\uparrow}(A|B)_{\rho} &\leq H_{2 - \frac{1}{\alpha}}^{\downarrow}(A|B)_{\rho}\,,\label{eq:ineq2}\\
    \widetilde{H}_{\alpha}^{\downarrow}(A|B)_{\rho} \leq \widetilde{H}_{\alpha}^{\uparrow}(A|B)_{\rho} &\leq \widetilde{H}_{2-\frac{1}{\alpha}}^{\downarrow}(A|B)_{\rho}\,,\label{eq:ineq3}\\
    H_{\alpha}^{\downarrow}(A|B)_{\rho} \leq \widetilde{H}_{\alpha}^{\downarrow}(A|B)_{\rho} &\leq H_{2-\frac{1}{\alpha}}^{\downarrow}(A|B)_{\rho}\label{eq:ineq4}\,. 
  \end{align}
\end{corollary}
\begin{proof}
  Note that the first inequality on each line follows directly from the relations depicted in Figure~\ref{fig:overview}.
  Next, consider an arbitrary purification $\rho_{ABC} \in \cS(ABC)$ of $\rho_{AB}$. The relations of Figure~\ref{fig:overview}, for any $\gamma \geq 0$, applied to the marginal $\rho_{AC}$ are given as
  \begin{align}
    &\widetilde{H}_{\gamma}^{\uparrow}(A|C)_{\rho} \geq \widetilde{H}_{\gamma}^{\downarrow}(A|C)_{\rho} \geq H_{\gamma}^{\downarrow}(A|C)_{\rho}\,, \qquad \textrm{and} \\
    &\widetilde{H}_{\gamma}^{\uparrow}(A|C)_{\rho} \geq H_{\gamma}^{\uparrow}(A|C)_{\rho} \geq H_{\gamma}^{\downarrow}(A|C)_{\rho}\,.
  \end{align}
  We then substitute the corresponding dual entropies according to Theorem~\ref{thm:dual}, which yields the desired inequalities upon appropriate new parametrization.
\end{proof}

We note that the fully classical (commutative) case of all these inequalities is trivial except for the second inequalities in~\eqref{eq:ineq2} and~\eqref{eq:ineq3}, which were proven before by one of authors~\cite[Lem. 6]{Hayashi:2013aa}. Other special cases of these inequalities are also well known and have operational significance. For example,~\eqref{eq:ineq3} for $\alpha = \infty$ states that $\widetilde{H}_{\infty}^{\uparrow}(A|B)_{\rho} \leq \widetilde{H}_{2}^{\downarrow}(A|B)_{\rho}$, which relates the conditional min-entropy in~\eqref{eq:min2} to the conditional collision entropy in~\eqref{eq:h2}. To understand this inequality more operationally we rewrite the conditional min-entropy as its dual semi-definite program~\cite{koenig08},
\begin{align}
\widetilde{H}_{\infty}^{\uparrow}(A|B)_{\rho}=\inf_{\Lambda_{B\rightarrow A'}}-\log\big(|A|\cdot F(\Phi_{AA'},\Lambda_{B\rightarrow A'}[\rho_{AB}]\big)\,,
\end{align}
where $A'$ is a copy of $A$, the infimum is over all quantum channels $\Lambda_{B\rightarrow A'}$, $|A|$ denotes the dimension of $A$, and $\Phi_{AA'}$ is the maximally entangled state on $AA'$. Now, the above inequality becomes apparent since the conditional collision entropy can be written as~\cite{berta13},
\begin{align}
\widetilde{H}_{2}^{\downarrow}(A|B)_{\rho}=-\log\big(|A|\cdot F(\Phi_{AA'},\Lambda_{B\rightarrow A'}^{\mathrm{pg}}[\rho_{AB}]\big)\,,
\end{align}
where $\Lambda_{B\rightarrow A'}^{\mathrm{pg}}$ denotes the pretty good recovery map of Barnum and Knill~\cite{Barnum02}. Also,~\eqref{eq:ineq1} for $\alpha = \frac12$ yields $\widetilde{H}_{\nicefrac{1}{2}}^{\uparrow}(A|B)_{\rho} \leq H_0^{\uparrow}(A|B)_{\rho}$, which relates the quantum conditional max-entropy in~\eqref{eq:hmax} to the quantum conditional generalization of the Hartley entropy in~\eqref{eq:h0}.

We believe that the sandwich relations~\eqref{eq:ineq1}--\eqref{eq:ineq4} for $\alpha$ close to $1$ will prove useful in applications in quantum information theory as they allow to switch between different definitions of the conditional R\'enyi entropy.

\section{Entropic Uncertainty Relations}

A series of papers~\cite{berta10,tomamichel11} culminating in~\cite{colbeck11} established a general technique to derive uncertainty relations for quantum conditional entropies based on two main ingredients: (1) a duality relation, and (2) a data-processing inequality for the underlying divergence. It is evident that all our definitions of conditional R\'enyi entropies fit the framework of~\cite{colbeck11}, which then immediately yields the following entropic uncertainty relations:
\begin{corollary}
  Let $\rho_{ABC} \in \cS(ABC)$ and let $\{ M_x \}_x$ and $\{ N_y \}_y$ be two positive operator-valued measures. We define \emph{the overlap} $c := \max_{x,y} \big\| \sqrt{M_x} \sqrt{N_y} \big\|$ and consider the \emph{post-measurement states}
  \begin{align}
    \rho_{XB} := \bigoplus_x \tr_{AC} \big\{ M_x \rho_{ABC} \big\} \quad \textrm{and} \quad \rho_{YC} := \bigoplus_y \tr_{AB} \big\{ N_y \rho_{ABC} \big\} \,. \label{eq:pm-states}
  \end{align}
 Then, the following relations hold:
  \begin{align}
     H_{\alpha}^{\downarrow}(X|B)_{\rho} + H_{\beta}^{\downarrow}(Y|C)_{\rho} \geq \log \frac{1}{c}, \qquad && \textrm{for}  &&\alpha, \beta \in [0, 2],\, \alpha + \beta = 2, \label{eq:ucr1}\\
          \widetilde{H}_{\alpha}^{\uparrow}(X|B)_{\rho} + \widetilde{H}_{\beta}^{\uparrow}(Y|C)_{\rho} \geq \log \frac{1}{c}, \qquad && \textrm{for}  &&\alpha, \beta \in \Big[\frac12, \infty\Big],\, \frac{1}{\alpha} + \frac{1}{\beta} = 2 \label{eq:ucr2},\\
     H_{\alpha}^{\uparrow}(X|B)_{\rho} + \widetilde{H}_{\beta}^{\downarrow}(Y|C)_{\rho} \geq \log \frac{1}{c}, \qquad && \textrm{for}  &&\alpha \in [0, 2],\, \beta \in \Big[\frac12, \infty\Big],\, \alpha \cdot \beta = 1. \label{eq:ucr3}
  \end{align}
\end{corollary}
We want to point out that the first and second inequality were first shown in~\cite{colbeck11} and~\cite{lennert13}, respectively; the third inequality is novel. To verify it, we apply~\cite[Thm.~1]{colbeck11} to $H_{\alpha}^{\uparrow}(X|B)_{\rho}$ and note that $H_{\alpha}^{\uparrow}(X|B)_{\rho}$ has the required form. Furthermore, it is already pointed out in~\cite{colbeck11} that the underlying divergence, $D_{\alpha}(\rho\|\sigma)$ for $\alpha \in [0, 2]$, satisfies the required properties for the application of their theorem.
As such, comparing~\eqref{eq:ucr3} to the corresponding duality relation~\eqref{eq:dual3}, we see that in order to derive the uncertainty relation we need to restrict to $\alpha \in [0,2]$ to be in the regime where data-processing holds.

It is noteworthy that even for the case of classical side information (if the systems $B$ and $C$ are classical), the three relations are genuinely different. The first inequality bounds the sum of two $\downarrow$-entropies, the second the sum of two $\uparrow$-entropies, and the third inequality the sum of a $\downarrow$- and an $\uparrow$-entropy.
Let us further specialize these inequalities for the case where both $B$ and $C$ are trivial. It was already noted in~\cite{lennert13} that~\eqref{eq:ucr2} specializes to the well-known Maassen-Uffink relation~\cite{maassen88}. We have
\begin{align}
  H_{\alpha}(X)_{\rho} + H_{\beta}(Z)_{\rho} \geq \log \frac1{c} \quad \textrm{for} \quad \alpha, \beta \in \Big[\frac12, \infty\Big],\, \frac{1}{\alpha} + \frac{1}{\beta} = 2,
\end{align}
evaluated for the marginals of the states in~\eqref{eq:pm-states}. It is also easy to verify that~\eqref{eq:ucr1} and~\eqref{eq:ucr3} specialize to strictly weaker uncertainty relations when $B$ and $C$ are trivial.

\paragraph*{Acknowledgments.}
MT is funded by the Ministry of Education (MOE) and National Research Foundation Singapore, as well as MOE Tier 3 Grant ``Random numbers from quantum processes'' (MOE2012-T3-1-009). MB thanks the Center for Quantum Technologies, Singapore, for hosting him while this work was done.
MH is partially supported by a MEXT Grant-in-Aid for Scientific Research (A) No.~23246071
and the National Institute of Information and Communication Technology (NICT), Japan.
The Centre for Quantum Technologies is funded by the
Singapore Ministry of Education and the National Research Foundation
as part of the Research Centres of Excellence programme.

\appendix

\section{H\"older Inequalities}
\label{app:hoelder}

We prove the following H\"older and reverse H\"older inequalities for traces of operators. 

\begin{lemma}
  \label{lm:hoelder}
  Let $A, B \geq 0$ and let $p > 0$, $q \in \mathbb{R}$ such that $\frac{1}{p} + \frac{1}{q} = 1$. Then, the following H\"older and reverse H\"older inequalities hold:
    \begin{align}
      \tr\{ A B \} &\leq \big( \tr \{ A^{p} \} \big)^{\frac{1}{p}} \big( \tr \{ B^{q} \} \big)^{\frac{1}{q}} \qquad \textrm{if }  p > 1 \,, \label{eq:hdirect}\\
      \tr\{ A B \} &\geq \big( \tr \{ A^{p} \} \big)^{\frac{1}{p}} \big( \tr \{ B^{q} \} \big)^{\frac{1}{q}} \qquad \textrm{if }  p < 1 \textrm{ and } B \gg A \,. \label{eq:hreverse}
    \end{align}
  Here, $B^q$ is evaluated on the support of $B$ by convention.
\end{lemma}

The first statement also immediately follows from a H\"older inequality for unitarily invariant norms (the trace norm in this case), e.g.~in~\cite[Cor.~IV.2.6]{bhatia97}. However, we believe that the following reduction of the proof to the commutative case is noteworthy.

\begin{proof}
  For commuting $A$ and $B$, the above result immediately follows from the corresponding classical H\"older and reverse H\"older inequalities. Now, let $\mathcal{M}$ be a pinching in the eigenbasis of $B$. Since $\mathcal{M}[A]$ commutes with $B$, we have
  \begin{align}
    \tr\{ A B \} = \tr\{\mathcal{M}[A] B\} &\leq \big( \tr \big\{ \big( \mathcal{M}[A] \big)^{p} \big\} \big)^{\frac{1}{p}} \big( \tr \{ B^{q} \} \big)^{\frac{1}{q}} \qquad \textrm{if }  p > 1, \\
      \tr\{ A B \} = \tr\{\mathcal{M}[A] B\} &\geq \big( \tr \big\{ \big( \mathcal{M}[A] \big)^{p} \big\} \big)^{\frac{1}{p}} \big( \tr \{ B^{q} \} \big)^{\frac{1}{q}}     \qquad \textrm{if }  p < 1 \,.
  \end{align}
under the respective constraints.
 Now, note that for $p > 1$, we have $\| \mathcal{M}[A] \|_p \leq \| A \|_p$ by the pinching inequality for the Schatten $p$-norm~\cite[Eq.~(IV.52)]{bhatia97} and~\eqref{eq:hdirect} follows. On the other hand, for $p < 1$, we use~\cite[Thm. V.2.1]{bhatia97}, which implies that $\big(\mathcal{M}[A]\big)^p \geq \mathcal{M}[A^p]$. This yields~\eqref{eq:hreverse} and concludes the proof.
\end{proof}


\end{document}